\documentclass[11pt,a4paper]{article}
\usepackage[margin=2.5cm]{geometry}
\usepackage{booktabs}
\usepackage{hyperref}
\usepackage{amsmath}
\usepackage{amssymb}
\usepackage{graphicx}
\usepackage{caption}
\usepackage{subcaption}
\usepackage{listings}
\usepackage{xcolor}
\usepackage[numbers]{natbib}
\usepackage{algorithm}
\usepackage{algpseudocode}
\usepackage{tikz}
\usetikzlibrary{decorations.pathreplacing,positioning}
\usepackage{amsthm}
\usepackage{parskip}
\usepackage{CJKutf8}

\theoremstyle{definition}
\newtheorem{definition}{Definition}
\theoremstyle{plain}
\newtheorem{lemma}{Lemma}
\newtheorem{theorem}{Theorem}
\newtheorem{proposition}{Proposition}

\definecolor{codegreen}{rgb}{0,0.6,0}
\definecolor{codegray}{rgb}{0.5,0.5,0.5}
\definecolor{codepurple}{rgb}{0.58,0,0.82}
\definecolor{backcolour}{rgb}{0.95,0.95,0.92}

\lstdefinestyle{mystyle}{
backgroundcolor=\color{backcolour},
commentstyle=\color{codegreen},
keywordstyle=\color{magenta},
numberstyle=\tiny\color{codegray},
stringstyle=\color{codepurple},
basicstyle=\ttfamily\footnotesize,
breakatwhitespace=false,
breaklines=true,
captionpos=b,
keepspaces=true,
numbers=left,
numbersep=5pt,
showspaces=false,
showstringspaces=false,
showtabs=false,
tabsize=2
}
\title{\textbf{The Li-Chao Tree: Algorithm Specification and Analysis}}
\author{Chao Li \\ \texttt{chnlich@live.com}}
\date{\today}

\begin{document}

\maketitle

\begin{abstract}
The Li-Chao tree (LICT) was first introduced in lecture materials~\cite{lichao2012} as an efficient data structure for dynamic lower envelope maintenance. In the years since, it has achieved widespread adoption within the competitive programming community, yet no formal specification has appeared in the peer-reviewed literature. This paper provides the definitive formalization of the Li-Chao tree, serving as both the official specification and an expansion of the original lecture materials~\cite{lichao2012}. We present complete algorithmic specifications, establish formal correctness proofs, analyze theoretical complexity, and provide empirical performance characterization. The LICT offers distinct advantages in implementation simplicity, numerical stability, and extensibility to advanced variants such as persistence and line segments.
\end{abstract}

\section{Introduction}

Dynamic Lower Envelope maintenance is a fundamental problem in computational geometry with extensive applications. We focus on the Li-Chao tree (LICT), a data structure that maintains a set of linear functions while supporting insertion and query operations. The problem is defined as follows: given a dynamic set of linear functions $y = kx + b$, support efficient insertion of new lines and querying the minimum (or equivalently, maximum) value at arbitrary $x$ coordinates. This report focuses on minimum queries; maximum queries are obtained by negating line parameters.

Formally, we require a data structure supporting two operations:
\begin{enumerate}
\item \textbf{Add Line:} Insert a new line $y = kx + b$ into the structure.
\item \textbf{Query:} Given $x_0$, compute
\[
\min_{i} \{k_i x_0 + b_i\}
\]
over all lines currently in the structure.
\end{enumerate}

The LICT provides $O(\log C)$ time per operation for minimum (or maximum) queries, where the complexity is parameterized by the coordinate universe size rather than the number of inserted lines.  Here,
\[
C = \frac{\text{coordinate range}}{\text{precision level}}
\]
represents the ratio of the coordinate range to the precision level.
For integer coordinates with range $[0, 10^9]$, $C = 10^9$; for the same range with precision $10^{-6}$, $C = 10^{15}$.

The structure also supports line segments (lines defined only on finite intervals $[x_l, x_r]$) in addition to infinite lines.

\section{Related Work}\label{sec:related}

Several approaches exist for Dynamic Lower Envelope maintenance, each with distinct trade-offs in terms of time complexity, implementation complexity, and applicability constraints. We review these solutions to establish the context for the LICT.

Dynamic maintenance of geometric configurations has been studied extensively in computational geometry.

\subsection{Overmars and van Leeuwen (1981)}

Overmars and van Leeuwen \cite{overmars1981} presented foundational work on dynamic convex hull maintenance. Their data structure supports insertion and deletion of lines while maintaining the lower envelope, enabling efficient querying of the minimum value at any point.

Their approach uses a balanced binary search tree to explicitly maintain the convex hull. Each node stores a line, and the tree is ordered by slope. Intersection points between adjacent lines are computed to determine the hull structure. Queries take $O(\log N)$ time, while insertions and deletions require $O(\log^2 N)$ time.

\subsection{Monotonic Convex Hull Trick}

The monotonic convex hull trick addresses the special case where line slopes are inserted in monotonically increasing or decreasing order~\cite{cpalgorithms_cht}.

When insertions arrive in order of monotonically increasing (or decreasing) slopes, a deque-based approach achieves $O(1)$ amortized time per insertion and $O(1)$ amortized time per query. This variant, widely used in dynamic programming optimization, maintains the convex hull incrementally without requiring balanced tree structures. However, the monotonicity restriction limits its applicability to problems where line slopes are known to follow a specific order.

\subsection{Dynamic Convex Hull Trick}

For arbitrary insertion sequences without monotonicity constraints, a balanced binary search tree maintains the hull explicitly~\cite{kactl}. Each insertion and query requires $O(\log N)$ amortized time. The implementation complexity stems from the need to compute and maintain intersection points between adjacent lines in the hull.

The Dynamic CHT requires computing intersection points between adjacent lines in the hull as
\[
x_{\text{intersect}} = \frac{b_2 - b_1}{k_1 - k_2},
\]
necessitating careful handling of precision (near parallel lines).

\section{The Li-Chao Tree}\label{sec:lichao}

This section presents the Li-Chao tree.

\subsection{Core Insight}

The Li-Chao tree is built upon a fundamental observation about line dominance over intervals. Consider two lines defined over an interval $[l, r]$:
\[
f_1(x) = k_1x + b_1 \quad \text{and} \quad f_2(x) = k_2x + b_2.
\]
Since two distinct lines intersect at most once, one line must be lower than the other on at least half of the interval. We store the line that wins at the midpoint $m = \frac{l+r}{2}$ and push the other line into the child that contains the intersection. The pushed line only requires consideration on at most half the original interval. This observation enables a recursive decomposition: either the query is answered by the stored line, or we recurse on a subproblem of strictly half the size.

This interval-halving property yields the logarithmic query time and forms the basis of the tree structure.

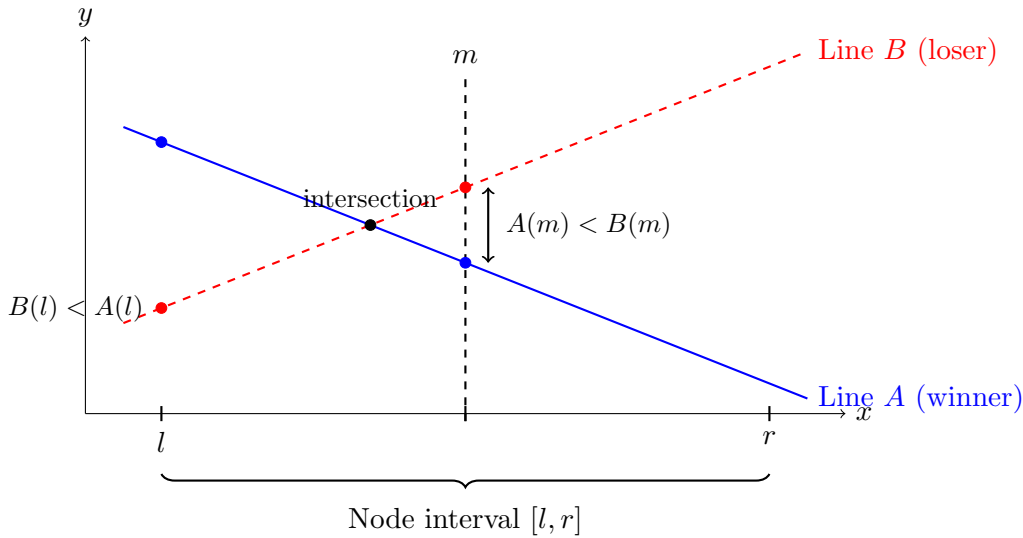
\begin{figure}[H]
\centering
\begin{tikzpicture}[scale=1.0]
\draw[->] (0,0) -- (10,0) node[right] {$x$};
\draw[->] (0,0) -- (0,5) node[above] {$y$};

\draw[thick] (1,0.1) -- (1,-0.1) node[below] {$l$};
\draw[thick] (9,0.1) -- (9,-0.1) node[below] {$r$};
\draw[thick, decorate, decoration={brace, amplitude=5pt, mirror}] (1,-0.8) -- (9,-0.8) node[midway, below=8pt] {Node interval $[l, r]$};

\draw[thick, dashed] (5,-0.1) -- (5,4.5) node[above] {$m$};
\draw[thick] (5,0.1) -- (5,-0.1);

\draw[thick, blue] (0.5,3.8) -- (9.5,0.2) node[right] {Line $A$ (winner)};
\draw[thick, red, dashed] (0.5,1.2) -- (9.5,4.8) node[right] {Line $B$ (loser)};

\filldraw[black] (3.75,2.5) circle (2pt);
\node[above] at (3.75,2.6) {\small intersection};

\filldraw[blue] (5,2.0) circle (2pt);
\filldraw[red] (5,3.0) circle (2pt);
\draw[thick, <->] (5.3,2.0) -- (5.3,3.0);
\node[right] at (5.4,2.5) {\small $A(m) < B(m)$};

\filldraw[blue] (1,3.6) circle (2pt);
\filldraw[red] (1,1.4) circle (2pt);
\node[left] at (0.9,1.4) {\small $B(l) < A(l)$};
\end{tikzpicture}
\caption{Interval Advantage Line Diagram. Line $A$ wins at the midpoint and is stored; Line $B$ (the loser) is routed to the left child where it may be optimal.}
\label{fig:interval-advantage}
\end{figure}

Note that this is \emph{local} optimality only. Lines routed down other branches may achieve lower values at the midpoint, so the stored line is not necessarily globally optimal at that coordinate.

\subsection{Algorithmic Approach}

\begin{definition}[Li Chao Tree]\label{def:lict}
A \emph{Li Chao Tree} over a domain $[L, R]$ is a binary tree where:
\begin{itemize}
  \item Each node owns an interval $[l, r] \subseteq [L, R]$ with midpoint $m = (l+r)/2$, and stores at most one line $f(x) = kx + b$.
  \item \textbf{Invariant.} The stored line achieves the minimum value at $m$ among all lines ever routed through that node.
  \item \textbf{Children.} Left child owns $[l, m]$; right child owns $[m, r]$. Children are created lazily on first insertion.
  \item \textbf{Initial state.} An empty tree is a single null root over $[L, R]$.
\end{itemize}
\end{definition}

The LICT offers a fundamentally different approach based on implicit envelope maintenance through interval subdivision. Rather than tracking the convex hull geometry explicitly, the structure maintains the best line at each interval, recursively partitioning the query range.

We first consider the simplest and most common setting: all queried $x$-coordinates are integers. In this case, if the coordinate range is $[0, C-1]$ for some positive integer $C$, the tree recursively bisects this range into subintervals $[l, r]$ where $l, r \in \mathbb{Z}$. Every leaf corresponds to a single integer, so the tree has depth $\lceil \log_2 C \rceil$ and every query lands at a unique leaf. This integer setting is the one most frequently encountered in competitive programming and in many algorithm design applications, and we use it throughout the examples below.

\textbf{Remark (non-integer queries).} The same structure handles queries at non-integer coordinates by refining the discretisation. If queries require precision $\epsilon$ (e.g., $\epsilon = 10^{-6}$), one can rescale the coordinate range by $1/\epsilon$, effectively treating each precision step as one unit. The coordinate universe size becomes $C = \text{coordinate range} / \epsilon$, and the tree depth grows to $\lceil \log_2 C \rceil$ accordingly. All algorithmic details remain identical; only the universe size changes. For example, a range of $[0, 10^9]$ with precision $10^{-6}$ yields $C = 10^{15}$ and tree depth $\approx 50$.

Each node in the tree represents an interval $[l, r]$ and stores one line.

\paragraph{Insertion.} A line is \emph{routed} through the tree following a single path from root to leaf. If the current node is empty, the new line is stored directly and routing terminates. Otherwise, two candidate lines compete at the node: the resident (currently stored) line and the incoming (new) line. The line that achieves the lower value at the midpoint
\[
m = \frac{l+r}{2}
\]
is designated the \emph{winner} and stored at the node; the \emph{loser} continues downward to the appropriate child subtree.

Since the winner is already determined at $m$, the remaining question is which half contains the intersection. If the loser is lower at $l$ (ordering at $l$ differs from $m$), the intersection lies in $[l, m]$, so the loser is routed to the left child. Otherwise, the intersection lies in $[m, r]$, and the loser is routed right.

When the insertion reaches maximum depth ($r - l < \epsilon$, where $\epsilon$ is the required precision), the interval is too narrow to distinguish any two query points. The stored line already resolves every query in that interval, so the losing line is simply dropped: no further recursion is needed. For integer coordinates ($\epsilon = 1$), this condition simplifies to $l = r$.

\begin{algorithm}[H]
\caption{LICT Line Insertion}
\label{alg:insert}
\begin{algorithmic}[1]
\Require Node pointer $node$, new line $new\_line$, interval $[l, r]$
\If{$node$ is null}
\State $node \gets \text{new Node}(new\_line)$
\State \Return
\EndIf
\State $m \gets l + (r - l) / 2$
\State $lef \gets new\_line.eval(l) < node.line.eval(l)$
\State $midf \gets new\_line.eval(m) < node.line.eval(m)$
\If{$midf$}
\State swap($node.line$, $new\_line$)
\EndIf
\If{$r - l < \epsilon$} \Comment{Max depth reached; drop the losing line}
\State \Return
\EndIf
\If{$lef \neq midf$}
\State \Call{Insert}{$node.left$, $new\_line$, $l$, $m$}
\Else
\State \Call{Insert}{$node.right$, $new\_line$, $m$, $r$}
\EndIf
\end{algorithmic}
\end{algorithm}

\textbf{Remark (tie-breaking).} When both lines achieve the same value at $m$, the strict $<$ comparison on line~8 keeps the resident line and pushes the incoming line down. This convention is arbitrary and preserves correctness. A duplicate line (identical $k$ and $b$) is a special case: whenever it meets an identical resident line, no swap occurs and the duplicate is routed to one child. The process then terminates either at the first null child (where the duplicate is stored) or at maximum depth (where it is dropped). In both cases, query answers are unchanged.

\begin{figure}[H]
\centering
\begin{tikzpicture}[scale=0.9, every node/.style={font=\small}]
\node[draw, rectangle, fill=blue!10, minimum width=2.2cm] (root) at (0,0) {$[0, 8]$: $f_0$};
\node[below=0.1cm of root] {\scriptsize $m=4$};
\node[draw, rectangle, fill=blue!10, minimum width=2.2cm] (l1) at (-2.5,-1.8) {$[0, 4]$: $f_1$};
\node[below=0.1cm of l1] {\scriptsize $m=2$};
\node[draw, rectangle, fill=blue!10, minimum width=2.2cm] (r1) at (2.5,-1.8) {$[4, 8]$: $f_2$};
\node[below=0.1cm of r1] {\scriptsize $m=6$};
\node[draw, rectangle, minimum width=1.4cm] (l2l) at (-4.0,-3.6) {$[0, 2]: -$};
\node[draw, rectangle, minimum width=1.4cm] (l2r) at (-1.0,-3.6) {$[2, 4]: -$};
\node[draw, rectangle, minimum width=1.4cm] (r2l) at (1.0,-3.6) {$[4, 6]: -$};
\node[draw, rectangle, minimum width=1.4cm] (r2r) at (4.0,-3.6) {$[6, 8]: -$};
\draw[->, thick] (root) -- (l1); \draw[->, thick] (root) -- (r1);
\draw[->] (l1) -- (l2l); \draw[->] (l1) -- (l2r);
\draw[->] (r1) -- (r2l); \draw[->] (r1) -- (r2r);
\end{tikzpicture}
\caption{Tree structure before inserting $f_{new}(x) = -x + 10$. Currently stores $f_0(x)=x$ at root, $f_1(x)=2x-4$ at $[0,4]$, and $f_2(x)=0.5x+2$ at $[4,8]$.}
\label{fig:insertion-example-before}
\end{figure}
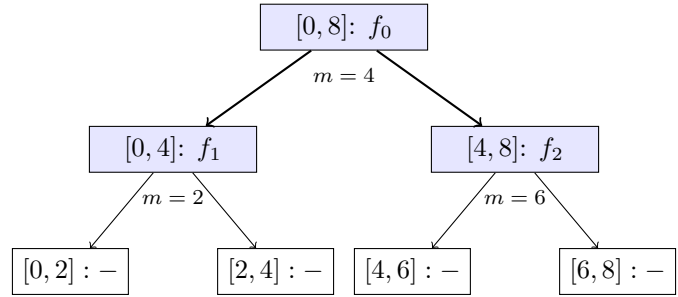

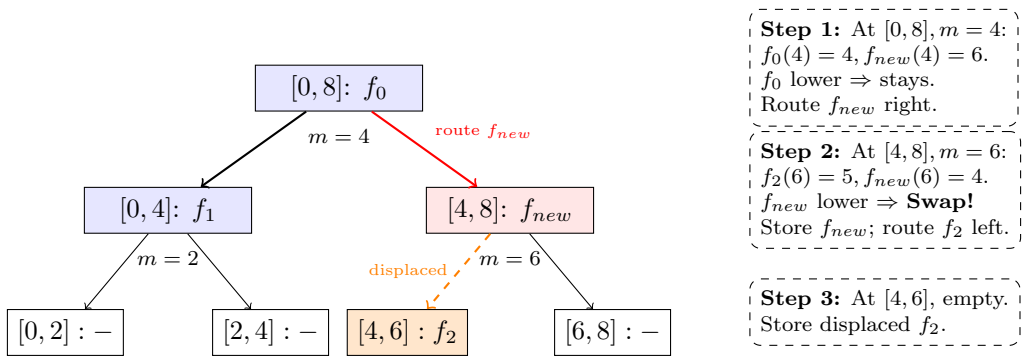
\begin{figure}[H]
\centering
\begin{tikzpicture}[scale=0.9, every node/.style={font=\small}]
\node[draw, rectangle, fill=blue!10, minimum width=2.2cm] (root) at (0,0) {$[0, 8]$: $f_0$};
\node[below=0.1cm of root] {\scriptsize $m=4$};
\node[draw, rectangle, fill=blue!10, minimum width=2.2cm] (l1) at (-2.5,-1.8) {$[0, 4]$: $f_1$};
\node[below=0.1cm of l1] {\scriptsize $m=2$};
\node[draw, rectangle, fill=red!10, minimum width=2.2cm] (r1) at (2.5,-1.8) {$[4, 8]$: $f_{new}$};
\node[below=0.1cm of r1] {\scriptsize $m=6$};
\node[draw, rectangle, minimum width=1.4cm] (l2l) at (-4.0,-3.6) {$[0, 2]: -$};
\node[draw, rectangle, minimum width=1.4cm] (l2r) at (-1.0,-3.6) {$[2, 4]: -$};
\node[draw, rectangle, fill=orange!20, minimum width=1.4cm] (r2l) at (1.0,-3.6) {$[4, 6]: f_2$};
\node[draw, rectangle, minimum width=1.4cm] (r2r) at (4.0,-3.6) {$[6, 8]: -$};
\draw[->, thick] (root) -- (l1);
\draw[->, thick, red] (root) -- (r1) node[midway, above right, font=\tiny] {route $f_{new}$};
\draw[->] (l1) -- (l2l); \draw[->] (l1) -- (l2r);
\draw[->, thick, dashed, orange] (r1) -- (r2l) node[midway, left, font=\tiny] {displaced};
\draw[->] (r1) -- (r2r);

\begin{scope}[xshift=6.0cm]
\node[draw, dashed, rounded corners, inner sep=4pt, align=left, anchor=west, fill=white, font=\scriptsize] at (0, 0.3) {
\textbf{Step 1:} At $[0,8], m=4$:\\
$f_0(4)=4, f_{new}(4)=6$.\\
$f_0$ lower $\Rightarrow$ stays.\\
Route $f_{new}$ right.
};
\node[draw, dashed, rounded corners, inner sep=4pt, align=left, anchor=west, fill=white, font=\scriptsize] at (0, -1.5) {
\textbf{Step 2:} At $[4,8], m=6$:\\
$f_2(6)=5, f_{new}(6)=4$.\\
$f_{new}$ lower $\Rightarrow$ \textbf{Swap!}\\
Store $f_{new}$; route $f_2$ left.
};
\node[draw, dashed, rounded corners, inner sep=4pt, align=left, anchor=west, fill=white, font=\scriptsize] at (0, -3.3) {
\textbf{Step 3:} At $[4,6]$, empty.\\
Store displaced $f_2$.
};
\end{scope}
\end{tikzpicture}
\caption{Tree structure after inserting $f_{new}(x) = -x + 10$. At $[4,8]$, $f_{new}(6)=4 < f_2(6)=5$, so $f_{new}$ swaps in and $f_2$ is displaced to $[4,6]$.}
\label{fig:insertion-example-after}
\end{figure}

\paragraph{Query.} We traverse the path to $x_0$, evaluating all stored lines along the path. The minimum value encountered equals the lower envelope at $x_0$ because any line that could be optimal at $x_0$ is stored on this path.

\begin{algorithm}[H]
\caption{LICT Query}
\label{alg:query}
\begin{algorithmic}[1]
\Require Node pointer $node$, query coordinate $x$, interval $[l, r]$
\Ensure Minimum value at coordinate $x$
\If{$node$ is null}
\State \Return $+\infty$
\EndIf
\State $m \gets l + (r - l) / 2$
\State $val \gets node.line.eval(x)$
\If{$r - l < \epsilon$}
\State \Return $val$
\EndIf
\If{$x \leq m$}
\State \Return $\min(val, \text{\Call{Query}{$node.left$, $x$, $l$, $m$}})$
\Else
\State \Return $\min(val, \text{\Call{Query}{$node.right$, $x$, $m$, $r$}})$
\EndIf
\end{algorithmic}
\end{algorithm}

\begin{figure}[H]
\centering
\begin{tikzpicture}[scale=0.8, every node/.style={font=\small}]
\node[draw, rectangle, fill=blue!10, minimum width=2.2cm] (root) at (0,0) {$[0, 8]$: $f_0$};
\node[below=0.1cm of root] {\scriptsize $m=4$};
\node[draw, rectangle, fill=blue!10, minimum width=2.2cm] (l1) at (-3.2,-2.0) {$[0, 4]$: $f_1$};
\node[below=0.1cm of l1] {\scriptsize $m=2$};
\node[draw, rectangle, fill=blue!10, minimum width=2.2cm] (r1) at (3.2,-2.0) {$[4, 8]$: $f_{new}$};
\node[below=0.1cm of r1] {\scriptsize $m=6$};
\node[draw, rectangle, minimum width=1.4cm] (l2ll) at (-5.2,-4.5) {$[0, 2]: -$};
\node[draw, rectangle, minimum width=1.4cm] (l2lr) at (-1.6,-4.5) {$[2, 4]: -$};
\node[draw, rectangle, fill=green!20, minimum width=1.4cm] (l2rl) at (1.6,-4.5) {$[4, 6]: f_2$};
\node[draw, rectangle, minimum width=1.4cm] (l2rr) at (5.2,-4.5) {$[6, 8]: -$};

\draw[->] (root) -- (l1);
\draw[->] (root) -- (r1);
\draw[->] (l1) -- (l2ll);
\draw[->] (l1) -- (l2lr);
\draw[->] (r1) -- (l2rl);
\draw[->] (r1) -- (l2rr);

\draw[->, very thick, blue] (root) -- (r1) node[midway, right, font=\tiny] {$5 > 4$};
\draw[->, very thick, blue] (r1) -- (l2rl) node[midway, left, font=\tiny] {$5 \le 6$};

\node[draw, thick, fill=white, rounded corners, align=left, font=\scriptsize, anchor=west] at (6.2, -1.8) {
\textbf{Query Result for $x=5$:}\\
1. Eval $f_0(5) = 5$\\
2. Eval $f_{new}(5) = 5$\\
3. Eval $f_2(5) = 4.5$\\
\textbf{Minimum:} $4.5$
};
\node[below=0.2cm of l2rl, red, font=\scriptsize\bfseries] {Target $x=5$};
\end{tikzpicture}
\caption{Query Tree View. To query at $x=5$, we follow the path $[0,8] \to [4,8] \to [4,6]$. We evaluate each stored line along the path ($f_0$, $f_{new}$, and $f_2$) and return the minimum.}
\label{fig:query-example}
\end{figure}
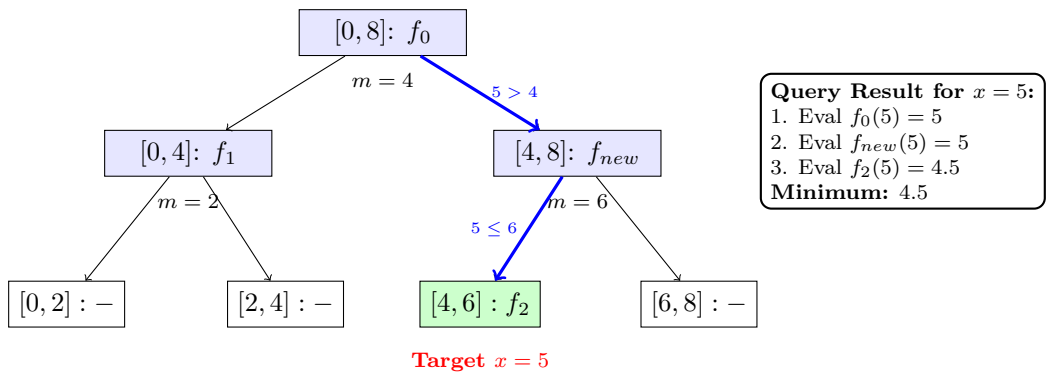

\subsection{Line Segment Extension}

The LICT extends naturally from full lines to line segments. To insert a line segment defined on $[x_l, x_r]$, we decompose this interval into $O(\log C)$ \emph{canonical intervals} of the implicit segment tree. Each canonical interval corresponds to a node whose range is entirely contained within $[x_l, x_r]$. Algorithm~\ref{alg:insert} is then called independently on each such interval to insert the line restricted to that subrange.

\begin{algorithm}[H]
\caption{LICT Segment Insertion}
\label{alg:segment-insert}
\begin{algorithmic}[1]
\Require Node pointer $node$, line $line$, node interval $[l, r]$, segment bounds $[x_l, x_r]$
\If{$x_r < l$ \textbf{or} $r < x_l$}
\State \Return \Comment{No overlap}
\EndIf
\If{$x_l \leq l$ \textbf{and} $r \leq x_r$}
\State \Call{Insert}{$node$, $line$, $l$, $r$} \Comment{Algorithm~\ref{alg:insert}}
\State \Return
\EndIf
\State $m \gets l + (r - l) / 2$
\State \Call{SegmentInsert}{$node.left$, $line$, $l$, $m$, $x_l$, $x_r$}
\State \Call{SegmentInsert}{$node.right$, $line$, $m$, $r$, $x_l$, $x_r$}
\end{algorithmic}
\end{algorithm}

The total cost per segment insertion is $O(\log^2 C)$: the interval $[x_l, x_r]$ decomposes into $O(\log C)$ canonical intervals, and each call to Algorithm~\ref{alg:insert} takes $O(\log C)$ time.

\subsection{Theoretical Analysis}

We now present a formal analysis of the LICT's correctness and complexity.

\begin{definition}[Coordinate Universe]
The LICT operates on a discrete domain where
\[
C = \frac{\text{coordinate range}}{\text{precision level}}.
\]
The tree depth is $h = \lceil \log_2 C \rceil$.
\end{definition}

\subsubsection{Correctness}

We establish the correctness of the LICT through three results: a foundational property of linear functions, a routing invariant maintained by insertion, and a global correctness theorem for queries.

\begin{lemma}[Linear Domination]\label{lem:linear-dom}
Let $f(x) = k_1 x + b_1$ and $g(x) = k_2 x + b_2$ be linear functions. If $f(a) \leq g(a)$ and $f(b) \leq g(b)$ for $a \leq b$, then $f(x) \leq g(x)$ for all $x \in [a, b]$.
\end{lemma}

\begin{proof}
Let $h(x) = f(x) - g(x)$. Since $h$ is linear with $h(a) \leq 0$ and $h(b) \leq 0$, for any $x \in [a, b]$:
\[
h(x) = h(a) \cdot \frac{b - x}{b - a} + h(b) \cdot \frac{x - a}{b - a} \leq 0. \qedhere
\]
\end{proof}

\begin{lemma}[Routing Dominance]\label{lem:routing-dom}
At a node with interval $[l, r]$ and midpoint $m$, if line $A$ is kept ($A(m) \leq B(m)$) and line $B$ is pushed to one child, then $A$ dominates $B$ on the opposite child's interval.
\end{lemma}

\begin{proof}
Without loss of generality, suppose $B$ is pushed to the left child, so $B(l) < A(l)$. Since $A(m) \leq B(m)$ and $B(l) < A(l)$, the two lines intersect in $[l, m]$. For all $x \in [m, r]$, the difference $A(x) - B(x)$ has the same sign as at $m$, so $A(x) \leq B(x)$. By Lemma~\ref{lem:linear-dom}, $A$ dominates $B$ on $[m, r]$.
\end{proof}

\begin{theorem}[Query Correctness]\label{thm:query-correct}
For any query point $x_0$, the value returned by the query operation equals $\min_{i} \{k_i x_0 + b_i\}$ over all lines in the structure.
\end{theorem}

\begin{proof}
The query traverses the unique root-to-leaf path $P$ containing $x_0$, returning the minimum of the stored lines' values at $x_0$. Since these are a subset of all inserted lines, $\text{query}(x_0) \geq \min_i L_i(x_0)$. It remains to show the reverse inequality.

Suppose for contradiction that $\text{query}(x_0) > \min_i L_i(x_0)$. Let $L^*$ be a line achieving the global minimum at $x_0$. Then no line stored on $P$ has value $\leq L^*(x_0)$ at $x_0$; in particular, $L^*$ is not stored on any node of $P$.

Since $L^*$ enters at the root (on $P$) and is not stored on $P$ in the final tree, at some node $w \in P$, $L^*$ was pushed to the child not containing $x_0$. By Lemma~\ref{lem:routing-dom}, the line $S$ kept at $w$ satisfies $S(x_0) \leq L^*(x_0)$, contradicting the assumption.

If $S$ is later displaced from $w$ by a new line $C$:
\begin{itemize}
\item If $S$ is routed toward $x_0$: $S$ moves deeper on $P$, preserving the contradiction.
\item If $S$ is routed away from $x_0$: by Lemma~\ref{lem:routing-dom}, $C(x_0) \leq S(x_0) \leq L^*(x_0)$, and $C$ is at $w \in P$. If $C$ is also displaced, repeat.
\end{itemize}
Since displacements are finite, the chain terminates with a contradiction.
\end{proof}

\begin{lemma}[Segment Insertion Correctness]\label{lem:segment-correct}
Algorithm~\ref{alg:segment-insert} correctly maintains the lower envelope for a line segment defined on $[x_l, x_r]$.
\end{lemma}

\begin{proof}
The interval $[x_l, x_r]$ decomposes into $O(\log C)$ canonical intervals of the implicit segment tree, each corresponding to a node whose range lies entirely within $[x_l, x_r]$. For each canonical interval $[l_i, r_i]$, Algorithm~\ref{alg:insert} inserts the line restricted to $[l_i, r_i]$. By Theorem~\ref{thm:query-correct} applied to the subtree rooted at $[l_i, r_i]$, each such insertion correctly maintains the lower envelope within $[l_i, r_i]$. Since the canonical intervals cover $[x_l, x_r]$ exactly and no insertion affects nodes outside $[x_l, x_r]$, the overall lower envelope is correctly maintained.
\end{proof}

\subsubsection{Complexity Analysis}

\begin{proposition}[Time Complexity]\label{prop:time-complexity}
Line insertion and query operations require $O(\log C)$ time. Line segment insertion (Algorithm~\ref{alg:segment-insert}) requires $O(\log^2 C)$ time.
\end{proposition}

\begin{proof}
The LICT is a binary tree of depth
\[
h = \lceil \log_2 C \rceil = O(\log C).
\]

For insertion: The new line follows exactly one root-to-leaf path. At each node, we perform: (1) $O(1)$ line evaluations ($kx + b$); (2) $O(1)$ comparisons; and (3) optionally a line swap. Each operation takes $O(1)$ time. With $O(\log C)$ nodes visited, total time is $O(\log C)$.

For query: The traversal follows one root-to-leaf path with $O(1)$ work per node (line evaluation and comparison). Total time is $O(\log C)$.

For segment insertion: The segment $[x_l, x_r]$ decomposes into $O(\log C)$ canonical intervals, each invoking Algorithm~\ref{alg:insert} at cost $O(\log C)$. The total time is $O(\log C) \times O(\log C) = O(\log^2 C)$.
\end{proof}

\begin{proposition}[Space Complexity]\label{prop:space-complexity}
The LICT stores at most $O(N)$ nodes in the worst case, where $N$ is the number of inserted lines.
\end{proposition}

\begin{proof}
Each full-line insertion follows a single root-to-leaf path. At each existing node, the algorithm performs $O(1)$ comparisons and recurses into exactly one child. If the path reaches a null child, it creates one new node and returns; if the path terminates at an existing leaf, no new node is created. Therefore each insertion creates at most one new node. With $N$ insertions, the total number of nodes is at most $N$, giving $O(N)$ space.

\textbf{Remark.} For line segment insertion (Algorithm~\ref{alg:segment-insert}), each segment decomposes into $O(\log C)$ canonical intervals, creating up to $O(\log C)$ nodes per insertion and $O(N \log C)$ nodes overall.
\end{proof}

\textbf{Deletion.} The standard LICT does not support efficient deletion. Removing a line would require traversing all nodes where that line might be stored and recomputing optimal lines from descendants, which requires $\Omega(N)$ time in the worst case. For applications requiring deletion, an alternative approach is reconstructing the tree periodically.

The LICT's complexity depends on the coordinate range rather than the number of lines. This provides consistent performance regardless of hull size, though it cannot exploit cases where the hull remains small relative to the number of insertions.

\section{Benchmarks}\label{sec:benchmarks}

Reference implementations and benchmarks are available at \url{https://github.com/chnlich/lichao-tree}.

To validate the theoretical complexity analysis and characterize practical performance differences between the LICT and Dynamic CHT, we conducted systematic empirical evaluation across varying problem scales and input distributions.

\subsection{Experimental Setup}

All benchmarks were conducted with the following configuration.

\textbf{Hardware Specifications:}
\begin{itemize}
\item \textbf{CPU:} AMD Ryzen 9 3950X 16-Core Processor @ 3.50GHz (Zen 2 microarchitecture)
\item \textbf{Core Configuration:} 16 cores, 32 threads (SMT enabled)
\item \textbf{L1 Cache:} 512 KiB L1 data cache, 512 KiB L1 instruction cache
\item \textbf{L2 Cache:} 8 MiB
\item \textbf{L3 Cache:} 16 MiB (shared)
\item \textbf{Memory:} 64 GB DRAM
\item \textbf{Platform:} Linux x86\_64 (WSL2 virtualized)
\end{itemize}

\textbf{Software Configuration:}
\begin{itemize}
\item \textbf{Compiler:} g++ 11.4.0
\item \textbf{Optimization Flags:} \texttt{-O3 -std=c++17 -Wall -Wextra}
\item \textbf{Operating System:} Ubuntu 22.04.5 LTS (WSL2, kernel 6.6.87)
\end{itemize}

The Dynamic CHT baseline uses the KACTL \texttt{LineContainer}~\cite{kactl}, a \texttt{std::multiset}-based implementation widely adopted in competitive programming.

\textbf{Experimental Parameters:}
\begin{itemize}
\item \textbf{Test sizes:} $10^5$, $10^6$, and $10^7$ operations
\item \textbf{Coordinate range:}
\[
C = 10^9
\]
(assuming integer precision $\epsilon = 1$), with $x, k, b \in [-10^9, 10^9]$
\item \textbf{Random seed:} 42
\item \textbf{Measurement protocol:} Each test was run 10 times; reported times are the average. Variance was low ($<5\%$ coefficient of variation across runs). Benchmarks were run on an isolated system with no other user processes to minimize timing noise.
\item \textbf{Distributions:} We write
\[
X \sim U(a, b)
\]
to denote that random variable $X$ is drawn from a continuous uniform distribution over the interval $[a, b]$.
\begin{itemize}
\item \textbf{Random:} Slopes $k \sim U(-10^9, 10^9)$, intercepts $b \sim U(-10^9, 10^9)$. Expected hull size:
\[
\Theta(\log N).
\]
By point-line duality~\cite{deBerg2008}, each line $y = kx + b$ maps to the dual point $(k, b)$. The lower envelope size of the lines equals the lower convex hull size of the dual point set, which is uniformly distributed in the rectangle $[-10^9, 10^9]^2$. By the R\'{e}nyi and Sulanke theorem~\cite{renyi1963}, the expected number of convex hull vertices of $N$ uniform random points in a convex polygon is $\Theta(\log N)$.
\item \textbf{All on Hull:} Approximately $N/2$ lines $y = -(i+1)x + (i+1)^2$ for $i \in [0, \lfloor N/2 \rfloor - 1]$, interleaved with approximately $N/2$ queries. All inserted lines contribute to the hull, since the dual points $(-(i+1),\,(i+1)^2)$ lie on the convex parabola $b = k^2$, so all are vertices of the lower hull.
\end{itemize}
\end{itemize}

\subsection{Results}

Table~\ref{tab:results} presents the performance comparison.

\begin{table}[h]
\centering
\caption{Performance Comparison (Time in milliseconds)}
\label{tab:results}
\begin{tabular}{@{}cccccc@{}}
\toprule
\textbf{N} & \textbf{Distribution} & \textbf{Algorithm} & \textbf{Insert (ms)} & \textbf{Query (ms)} & \textbf{Total (ms)} \\
\midrule
$10^5$ & Random & LICT & 4.20 & 0.78 & 4.98 \\
$10^5$ & Random & Dynamic CHT & 2.21 & 0.29 & 2.50 \\
$10^5$ & All on Hull & LICT & 9.17 & 7.25 & 16.42 \\
$10^5$ & All on Hull & Dynamic CHT & 7.88 & 7.00 & 14.88 \\
\midrule
$10^6$ & Random & LICT & 41.63 & 7.84 & 49.48 \\
$10^6$ & Random & Dynamic CHT & 22.74 & 2.80 & 25.53 \\
$10^6$ & All on Hull & LICT & 245.17 & 248.77 & 493.94 \\
$10^6$ & All on Hull & Dynamic CHT & 294.86 & 304.05 & 598.91 \\
\midrule
$10^7$ & Random & LICT & 431.80 & 78.88 & 510.68 \\
$10^7$ & Random & Dynamic CHT & 226.64 & 33.48 & 260.12 \\
$10^7$ & All on Hull & LICT & 6643.19 & 7045.61 & 13688.80 \\
$10^7$ & All on Hull & Dynamic CHT & 7558.58 & 7417.51 & 14976.10 \\
\bottomrule
\end{tabular}
\end{table}

\subsection{Parameter-Matched Comparison: The $N = C$ Regime}

The theoretical complexity analysis reveals that Dynamic CHT achieves $O(\log N)$ time while LICT achieves $O(\log C)$. To directly compare these regimes, we conduct experiments where the number of lines $N$ equals the coordinate universe size $C$. This configuration mirrors dynamic programming optimizations common in competitive programming, where the state recurrence takes the form:
\[
dp[i] = \min_{0 \le j < i} \{ dp[j] + \text{cost}(j, i) \}
\]
Here, $\text{cost}(j, i)$ represents a linear function of $i$. Since the queries (indices $i$) are within a fixed range $[0, N]$, the coordinate universe size $C$ effectively equals $N$.

In this static universe setting, an iterative segment tree~\cite{zkw2010} implementation (often called ZKW segment tree) can be employed to eliminate recursion overhead and may improve cache locality due to its contiguous memory layout. We therefore include a third variant, \textbf{ZKW LICT}, in this comparison to evaluate the benefits of this specialized optimization against the standard LICT and Dynamic CHT.

\begin{itemize}
\item \textbf{Configuration:} $N = C = 10^5$, $N = C = 10^6$, and $N = C = 10^7$
\item \textbf{Lines (Random):} $k, b \sim U(-C/2, C/2)$
\item \textbf{Lines (All on Hull):} $y = -(i+1)x + (i+1)^2$ for $i \in [0, N/2 - 1]$, shuffled before insertion
\item \textbf{Query points:} $x \sim U(-C/2, C/2)$ (Random), $x \sim U(0, N)$ (All on Hull)
\item \textbf{Operation mix:} $N/2$ insertions followed by $N/2$ queries
\end{itemize}

\begin{table}[h]
\centering
\caption{Performance in $N = C$ Regime (Time in milliseconds)}
\label{tab:nc-regime}
\begin{tabular}{@{}cccccc@{}}
\toprule
\textbf{N = C} & \textbf{Distribution} & \textbf{Algorithm} & \textbf{Insert (ms)} & \textbf{Query (ms)} & \textbf{Total (ms)} \\
\midrule
$10^5$ & Random & LICT & 2.21 & 0.94 & 3.16 \\
$10^5$ & Random & ZKW LICT & 2.40 & 0.58 & 2.98 \\
$10^5$ & Random & Dynamic CHT & 2.23 & 0.41 & 2.64 \\
$10^5$ & All on Hull & LICT & 7.11 & 5.42 & 12.54 \\
$10^5$ & All on Hull & ZKW LICT & 5.32 & 1.70 & 7.02 \\
$10^5$ & All on Hull & Dynamic CHT & 7.17 & 6.83 & 14.00 \\
\midrule
$10^6$ & Random & LICT & 26.54 & 9.46 & 36.00 \\
$10^6$ & Random & ZKW LICT & 27.52 & 6.59 & 34.11 \\
$10^6$ & Random & Dynamic CHT & 23.57 & 3.77 & 27.33 \\
$10^6$ & All on Hull & LICT & 127.28 & 177.07 & 304.35 \\
$10^6$ & All on Hull & ZKW LICT & 71.79 & 41.68 & 113.47 \\
$10^6$ & All on Hull & Dynamic CHT & 187.77 & 260.54 & 448.31 \\
\midrule
$10^7$ & Random & LICT & 311.44 & 99.31 & 410.75 \\
$10^7$ & Random & ZKW LICT & 324.54 & 77.98 & 402.51 \\
$10^7$ & Random & Dynamic CHT & 235.24 & 45.01 & 280.25 \\
$10^7$ & All on Hull & LICT & 4022.59 & 5248.88 & 9271.48 \\
$10^7$ & All on Hull & ZKW LICT & 1685.15 & 1428.63 & 3113.78 \\
$10^7$ & All on Hull & Dynamic CHT & 6163.38 & 6579.40 & 12742.78 \\
\bottomrule
\end{tabular}
\end{table}

\subsection{Analysis}

When $C$ is larger than $N$, the LICT exhibits comparable performance to Dynamic CHT despite its theoretically larger $O(\log C)$ complexity. The absolute difference between $\log C$ and $\log N$ remains modest (e.g., $\log(10^9) \approx 30$ vs $\log(10^7) \approx 23$), so the LICT's smaller constant factor, arising from division-free comparisons and the absence of tree rebalancing, compensates for the additional tree depth.

In the $N = C$ regime, the ZKW LICT variant demonstrates substantial gains over both alternatives on dense hull inputs, with the $N = C = 10^7$ all-on-hull benchmark completed approximately $4\times$ faster than Dynamic CHT. 

The LICT's division-free operations ensure consistent numerical accuracy across all inputs. For applications involving floating-point coordinates or extreme value ranges, the LICT provides robustness guarantees that the Dynamic CHT lacks, by avoiding intersection point computation.

\section{Conclusion}\label{sec:discussion}

The LICT is the preferred choice in the following scenarios:

\subsection{Advantages}

\textbf{Line segment support required.} When the problem involves line segments (lines valid only on subranges) rather than infinite lines, the LICT provides natural $O(\log^2 C)$ insertion. The Dynamic CHT can support segments but requires significantly more complex machinery.

\textbf{Persistence required.} Path copying~\cite{driscoll1989} in the LICT is straightforward: insertions modify only nodes along a single root-to-leaf path, so copying those nodes creates a new version sharing unmodified subtrees with the previous version. Achieving persistence in the Dynamic CHT is substantially more complex due to the need to maintain hull invariants across versions.

\textbf{Near parallel lines.} When working on lines with subtle slope differences, the LICT's division-free operations avoid numerical instability.

\textbf{Implementation time constraints.} In settings such as competitive programming where implementation speed matters, the LICT's simplicity offers a clear advantage. The reduced code size and elimination of geometric corner cases allow for faster, more confident implementation.

\subsection{Limitations}

The LICT has several limitations that affect its applicability:

\textbf{No efficient deletion.} As noted above, the standard LICT does not support deletion of individual lines. Removing a line would require traversing all nodes where that line might be stored and recomputing optimal lines from descendants, requiring $\Omega(N)$ time in the worst case.

\textbf{Coordinate range dependency.} The LICT's complexity depends on $C$, the coordinate range divided by precision. For very large coordinate ranges with fine precision (e.g., 64-bit floating-point values spanning the entire representable range), $C$ can become impractically large.

\subsection{Future Work}

Several directions for future research and development remain:

\textbf{Deletion support.} Developing an efficient deletion mechanism for the LICT would extend its applicability to dynamic scenarios requiring removal of lines. Potential approaches include lazy deletion with periodic reconstruction or augmenting nodes with additional structure to support efficient line removal.

\textbf{Cache-efficient variants.} The LICT's pointer-based tree structure exhibits poor cache locality compared to array-based representations. Investigating cache-oblivious or cache-aware variants could improve practical performance without sacrificing asymptotic guarantees.

\textbf{Parallel implementations.} The LICT's tree structure naturally supports parallel queries, but insertions are inherently sequential. Developing concurrent LICT variants that support parallel insertions while maintaining correctness would benefit multi-core applications.

\textbf{Higher-dimensional extensions.} While the LICT extends naturally to higher dimensions (maintaining hyperplanes instead of lines), the space and time complexities grow exponentially with dimension. Investigating approximate variants or dimensionality reduction techniques could broaden applicability.

\section*{Acknowledgments}
This paper was prepared with the assistance of Claude Code by Anthropic and CharlieBot (\url{https://github.com/chnlich/charlie-bot}), a multi-session agent framework with integrated \TeX{} editing support.

The author thanks the following community resources for their role in disseminating the Li-Chao tree to a wider audience: CP-Algorithms (\url{https://cp-algorithms.com/geometry/li_chao_tree.html}), the Codeforces tutorial by I\_LOVE\_TIGER (\url{https://codeforces.com/blog/entry/51275}), the KTH Algorithm Competition Template Library (KACTL, \url{https://github.com/kth-competitive-programming/kactl}), and OI Wiki (\url{https://oi-wiki.org/ds/li-chao-tree/}).

\bibliographystyle{unsrtnat}
\bibliography{references}

@misc{lichao2012,
  author       = {Li, Chao},
  title        = {\begin{CJK}{UTF8}{gbsn}线段树\end{CJK} [Segment Trees]},
  howpublished = {Lecture slides, Hangzhou Xuejun High School, China},
  year         = {2012},
  note         = {Available at \url{https://github.com/chnlich/lichao-tree/blob/main/doc/slides.ppt}}
}

@article{overmars1981,
  author  = {Overmars, M.H. and van Leeuwen, J.},
  title   = {Maintenance of configurations in the plane},
  journal = {Journal of Computer and System Sciences},
  volume  = {23},
  number  = {2},
  pages   = {166--204},
  year    = {1981}
}

@Inbook{deBerg2008,
author="de Berg, Mark
and Cheong, Otfried
and van Kreveld, Marc
and Overmars, Mark",
title="Delaunay Triangulations",
bookTitle="Computational Geometry: Algorithms and Applications",
year="2008",
publisher="Springer Berlin Heidelberg",
address="Berlin, Heidelberg",
pages="191--218",
isbn="978-3-540-77974-2",
doi="10.1007/978-3-540-77974-2_9",
url="https://doi.org/10.1007/978-3-540-77974-2_9"
}

@article{driscoll1989,
  author  = {Driscoll, J.R. and Sarnak, N. and Sleator, D.D. and Tarjan, R.E.},
  title   = {Making data structures persistent},
  journal = {Journal of Computer and System Sciences},
  volume  = {38},
  number  = {1},
  pages   = {86--124},
  year    = {1989}
}

@misc{zkw2010,
  author       = {Zhang, Kunwei},
  title        = {\begin{CJK}{UTF8}{gbsn}统计的力量\end{CJK} [The Power of Statistics]},
  howpublished = {Lecture slides, Tsinghua University, China},
  year         = {2010},
  note         = {Introduced the non-recursive (iterative) segment tree, widely known as the ZKW segment tree in the Chinese competitive programming community}
}

@misc{cpalgorithms_cht,
  title        = {Convex Hull Trick and Li Chao Tree},
  howpublished = {CP-Algorithms},
  url          = {https://cp-algorithms.com/geometry/convex_hull_trick.html},
  year         = {2024}
}

@misc{kactl,
  title        = {{KTH Algorithm Competition Template Library (KACTL)}},
  howpublished = {GitHub},
  url          = {https://github.com/kth-competitive-programming/kactl},
  year         = {2024}
}

@article{renyi1963,
  author  = {R{\'{e}}nyi, A. and Sulanke, R.},
  title   = {\"{U}ber die konvexe {H}\"{u}lle von $n$ zuf\"{a}llig gew\"{a}hlten {P}unkten},
  journal = {Zeitschrift f\"{u}r Wahrscheinlichkeitstheorie und verwandte Gebiete},
  volume  = {2},
  number  = {1},
  pages   = {75--84},
  year    = {1963}
}

\end{document}